\documentclass[onecolumn,10pt]{IEEEtran}   	
\usepackage{graphicx}				
\usepackage{amsmath}								
\usepackage{amssymb}
\usepackage{bm}
\usepackage{amsthm}
\usepackage{algpseudocode}
\usepackage{setspace}
\usepackage{algorithm}

\usepackage{graphicx}
\usepackage{caption}
\usepackage{subcaption}

\newtheorem{definition}{Definition}
\newtheorem{lemma}{Lemma}
\newtheorem{theorem}{Theorem}

\newtheorem{construction}{Construction}

\allowdisplaybreaks[3]

\title{Generic Secure Repair for Distributed Storage}
\author{\IEEEauthorblockN{ \textbf{Wentao Huang} and \textbf{Jehoshua Bruck}}\\
	\IEEEauthorblockA{\emph{California Institute of Technology, Pasadena, USA}\\
		\{whuang,bruck\}@caltech.edu\ \
	}
}						

\begin{document}
\maketitle
\begin{abstract}
This paper studies the problem of repairing secret sharing schemes, i.e., schemes that encode a message into $n$ shares, assigned to $n$ nodes, so that any $n-r$ nodes can decode the message but any colluding $z$ nodes cannot infer any information about the message. In the event of node failures so that shares held by the failed nodes are lost, the system needs to be repaired by reconstructing and reassigning the lost shares to the failed (or replacement) nodes. This can be achieved trivially by a trustworthy third-party that receives the shares of the available nodes, recompute and reassign the lost shares. The interesting question, studied in the paper, is how to repair without a trustworthy third-party. The main issue that arises is repair security: how to maintain the requirement that any colluding $z$ nodes, including the failed nodes, cannot learn any information about the message, during and after the repair process? We solve this secure repair problem from the perspective of secure multi-party computation. Specifically, we design generic repair schemes that can securely repair \emph{any} (scalar or vector) linear secret sharing schemes. We prove a lower bound on the repair bandwidth of secure repair schemes and show that the proposed secure repair schemes achieve the optimal repair bandwidth up to a small constant factor when $n$ dominates $z$, or when the secret sharing scheme being repaired has optimal rate. We adopt a formal information-theoretic approach in our analysis and bounds. A main idea in our schemes is to allow a more flexible repair model than the straightforward one-round repair model implicitly assumed by existing secure regenerating codes.
Particularly, the proposed secure repair schemes are simple and efficient two-round protocols. 
\end{abstract}

\section{Introduction}
The problem of repairing secret sharing schemes has attracted significant interests recently. Specifically, a  secret sharing scheme encodes a message into $n$ shares, such that the message can be decoded from any $n-r$ shares (reliability), and that any $z$ shares are independent of the message (security). In the setting of distributed storage, a system consists of $n$ nodes and one share is assigned to each node. Therefore a secret sharing scheme can tolerate $r$ node failures (erasures) as well as $z$ colluding adversarial nodes trying to infer information about the message. In the event of node failures, the shares held by the failed nodes are lost and in order to maintain the same level of reliability, the system needs to repair the failures by reconstructing the lost shares and reassigning them to the failed (or replacement) nodes. Two problems arise during the repair process, namely, 1) bandwidth efficiency: it is desirable to minimize the amount of communication induced by the repair process; 2) repair security: the system needs to maintain the  security requirement that any colluding $z$ nodes, including the failed (or replacement) nodes, cannot infer any information about the message, during and after the repair process.

Secure regenerating codes, e.g., \cite{Pawar11,Shah:2011bz,Rawat:2014cu,Koyluoglu14,Kadhe14,Tandon16,HuangKun17}, are a class of secret sharing schemes that are carefully designed to address the above problems. We classify secure regenerating codes into two categories: codes that only address the bandwidth efficiency problem (i.e., codes with non-secure repair), and codes that address both the bandwidth efficiency and the repair security problems (i.e., codes with secure repair). Specifically, codes with non-secure repair focus on reducing the repair bandwidth without worrying about the security of the repair process. For example, the codes that tolerate Type-I adversary in \cite{Tandon16} and the codes in \cite{Huang17} belong to this category. For this case one can think of having a trustworthy repair dealer that will receive information from the available helper nodes, reconstruct the lost share and then forward it to the failed node. The repair dealer may receive enough information to gain knowledge of the message, and therefore has to be trustworthy. In comparison, regenerating codes with secure repair guarantee by code design that such a dealer will not  learn any information about the message. This in fact removes the need for the dealer to be trustworthy and the failed node can act as the dealer. Unfortunately, the guarantee that the dealer cannot learn any information about the message is shown to come at a high cost in rate \cite{Pawar11,Tandon16}, because more independent randomness (keys) is required in order to protect the message from the dealer, resulting in increased overhead. Therefore, codes with non-secure repair in general have a significantly better rate and repair bandwidth (when normalized by rate) than codes with secure repair.
 
In this paper we address the problem of repair security from a different perspective, without needing to take the heavy penalty in rate and other aspects of efficiency as in the case of secure regenerating codes. The key idea is that we allow a more flexible repair protocol: secure regenerating codes implicitly assume a simple ``one-round'' repair protocol, in which the helper nodes transmit information to the failed nodes but they themselves do not receive information from other nodes. This implicit ``one-round'' assumption is expensive in terms of efficiency. We show that, just by slightly relaxing this assumption and allowing a ``two-round'' protocol, it becomes possible to securely repair \emph{any} secret sharing scheme in a black-box manner, in the sense that the proposed repair protocol is generic and there is no need to design or modify the secret sharing scheme. Refer to Figure \ref{fig:1} for a simple example of the two-round secure repair protocol.

We remark that a two-round protocol is advantageous in that more nodes are allowed to receive information rather than only the failed node. This is intuitively beneficial because, if $d>z$ nodes can receive information, then we can take advantage of the gap between $d$ and $z$ in the following way. During the repair process, let the information received by any $z$ nodes be independent randomness (so that the security requirement is met), and let the information received by all $d$ nodes reveal useful information on the lost share. We then use an extra round of communication to transmit the information on and only on the lost share from the $d$ nodes to the failed node so that the lost share can be reconstructed. Loosely speaking, we can think of the repair process as letting the failed node ``compute'' its share securely, so that it only learns the share but nothing else. This is naturally related to the problem of secure multi-party computation and the ideas in \cite{BenOr88,Chaum88} play an important role in our repair schemes. We remark that we adopt a formal information-theoretic approach in our analysis and bounds, which differs from many existing works on secure multi-party computation.
We also note that relaxing the repair process to involve more than one round is practical. For example, POTSHARDS \cite{Storer09} employs a heuristic multi-round repair scheme to improve the security of the repair process.


Our generic secure repair schemes have two important advantages over secure regenerating codes with secure repair. First, the generic nature implies that there is no need to compromise the efficiency of the secret sharing scheme for secure repair. Here, aspects of efficiency at stake are not limited to the rate and repair bandwidth discussed earlier, but also include, for example, computational complexity \cite{HB16b} and decoding bandwidth \cite{HB16}, as it is not clear how to construct secure regenerating codes with optimal computation or decoding bandwidth. Second, most secure regenerating codes focus on secure repair by a fixed number of helper nodes. In the case that not enough helper nodes are available due to multiple node failures, it is not clear how secure repair can be achieved.

We briefly summarize the contributions of the paper. In Section \ref{sec:2}, we present a generic two-round secure repair scheme based on the ideas in \cite{BenOr88,Chaum88}. Specifically, in the first round each helper node encodes its share into $z+1$ pieces using a secret sharing scheme, so that any $z$ pieces reveal no information about the share and that the share can be decoded from $z+1$ pieces. The $z+1$ pieces are sent to $z+1$ receiver nodes, and each receiver node receives a piece from each helper node (if the helper node and the receiver node are the same node, then the corresponding piece needs not be transmitted). For example, in Figure \ref{fig:1}-(b), the helper nodes and receiver nodes are both Nodes 2 and 3. The set of pieces received by all receiver nodes contains enough information to decode the shares of all helper nodes and the lost share. We then need to communicate the information about the lost share, but no extra information about the shares of the helper nodes, to the failed node. To achieve this, each receiver node locally computes a function that takes the pieces received by the node as inputs, and outputs a ``distilled'' piece such that the set of ``distilled'' pieces only contains information about the lost share. This set is then transmitted from the receiver nodes to the failed node. Refer to Figure \ref{fig:1}-(c) for an example.

The generic repair scheme in Section \ref{sec:2} requires a relatively large repair bandwidth. In Section \ref{sec:3}, we reduce the repair bandwidth of the scheme significantly by adopting the idea of parallelism in \cite{Franklin92}. Instead of repairing one single share at a time, we repair multiple shares together in parallel,  therefore amortizing the communication overhead over the multiple shares. This is achieved by letting all $n$ nodes be receiver nodes (instead of $z+1$ nodes) and by using a secret sharing scheme of a higher rate in the first round. The larger gap between the number of receiver nodes and $z$ implies that we can encode more information in the secret sharing scheme (so that it has a higher rate) and can repair more shares in parallel.

The generic repair schemes in Sections \ref{sec:2} and \ref{sec:3} can securely repair any \emph{scalar} linear secret sharing schemes. A more general class of schemes are \emph{vector} linear secret sharing schemes. For a vector linear scheme over a field, each node stores multiple elements of the field instead of a single element as in the scalar linear case. Many efficient secret sharing schemes, e.g., schemes with efficient decoding bandwidth \cite{HB16,Huang17}, schemes with efficient computation \cite{HB16b,HB17sRAID}, and schemes with efficient repair bandwidth \cite{Shah:2011bz,Huang17}, are intrinsically vector linear. In Section \ref{sec:4} we generalize our secure repair schemes to generically repair any vector linear schemes. In particular, this generalization allows us to leverage the property of secret sharing schemes with efficient (non-secure) repair bandwidth, i.e., secure regenerating codes with non-secure repair, to further reduce the (secure) repair bandwidth.

Finally, in Section \ref{sec:5} we prove an information-theoretic lower bound on the repair bandwidth of secure repair schemes. The bound implies that the secure repair schemes in Sections \ref{sec:3} and \ref{sec:4} achieve the optimal repair bandwidth within a small constant factor when $n$ dominates $z$, or when the secret sharing scheme being repaired has optimal rate.



\begin{figure}[tp]
	\centering
	\begin{subfigure}[b]{1\textwidth}
		\centering
		\includegraphics[width=0.33\textwidth]{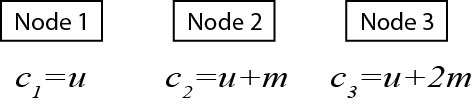}
		\caption{A secret sharing scheme over $\mathbb{F}_5$ with $r=1$ and $z=1$, where $m$ is a message symbol and $u$ is a random key uniformly distributed over $\mathbb{F}_5$. We denote the three shares by $c_1, c_2$ and $c_3$.}
	\end{subfigure}%
	\\
	\vspace{5mm}
	\begin{subfigure}[b]{1\textwidth}
		\centering
		\includegraphics[width=0.30\textwidth]{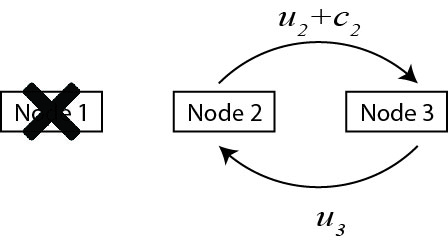}
		\caption{Repairing Node 1 (round 1): Node 2 generates a random symbol $u_2$ and sends $u_2+c_2$ to Node 3. Node 3 generates a random symbol $u_3$ and sends it to Node 2.}
	\end{subfigure}
	\\
	\vspace{5mm}
\begin{subfigure}[b]{1\textwidth}
	\centering
	\includegraphics[width=0.30\textwidth]{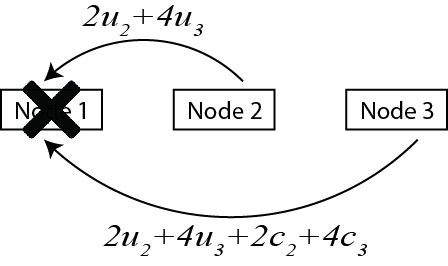}
	\caption{Repairing Node 1 (round 2): Node 2, having access to $u_2$ and $u_3$, computes and sends $2u_2 + 4u_3$ to Node 1. Node 3, having access to $u_2+c_2$, $u_3$ and $c_3$, computes and sends $2 u_2 + 4u_3 + 2c_2 + 4 c_3$ to Node 1. Node 1 can reconstruct its share since $c_1 = 2c_2 + 4c_3 $.}
\end{subfigure}
	\caption{Securely repairing a secret sharing scheme. Note that it is impossible to securely repair any node failure under the one-round repair model of regenerating codes, because for the failed node to reconstruct its share it has to collect the shares from the other two nodes, which will violate the security requirement. However, any node failure can be securely repaired by the two-round scheme shown above. To see that the scheme is secure, note that after the repair process Node 1 has access to $c_1$ and $2u_2 + 4u_3$; Node 2 has access to $c_2$, $u_2$ and $u_3$; Node 3 has access to $c_3$, $u_3$ and $u_2+c_2$. Therefore, any single node has access to only one share as well as some random symbols that are independent of the shares. Therefore no single node can learn any information about the message $m$.}
	\label{fig:1}
\end{figure}

\section{Generic secure repair} \label{sec:2}
Secret sharing schemes address the problem of storing a secret message securely and reliably. Specifically, an  $(n,k,r,z)$ secret sharing scheme over $\mathbb{F}_q$ is a randomized function that maps (encodes) a message $\bm{m}=(m_1,\cdots,m_k)$ of $k$ symbols over $\mathbb{F}_q$ to $n$ shares $\bm{c}=(c_1,\cdots,c_n)$ over $\mathbb{F}_q$, such that  1) $\bm{m}$ can be decoded from any subset of $ n-r$ shares; 2) any subset of $z$ shares do not reveal  information on $\bm{m}$.  Shamir's scheme is a well known secret sharing scheme with $k=1$.
\begin{construction} \label{con:sha}
	\emph{(Shamir's scheme \cite{Shamir:1979vo})} For any $n$, and $z<n$, let $k=1$, $r=n-z-1$ and $\mathbb{F}_q$ be a finite field of size $q>n$. Let $u_i$, $i \in [z]$ be i.i.d. uniformly distributed over $\mathbb{F}_q$ (also referred to as keys) and let $\alpha_i$, $i\in[n]$ be arbitrary distinct non-zero elements of $\mathbb{F}_q$. The shares corresponding to message $m_1$ are
	\begin{align} \label{eq:shamir}
	 (c_1, c_2, \cdots,c_n) = (m_1, u_1, u_2, \cdots, u_z) \left[ \begin{array}{cccc}
	 1 & 1 & \cdots & 1\\
	 \alpha_1 & \alpha_2 & \cdots & \alpha_n\\
	 \vdots & \vdots & \vdots & \vdots\\
	 \alpha_1^{z} & \alpha_2^{z} & \cdots & \alpha_n^{z}
	 \end{array} \right].
 	\end{align}
\end{construction}

\begin{lemma}\label{lem:linear}
	Let $c_i$, $i \in [n]$ be the shares of Shamir's scheme (\ref{eq:shamir}) on message $m_1$ and keys $u_i$, $i\in[z]$, and let $c'_i$, $i \in [n]$ be the shares of the scheme on message $m'_1$ and keys $u'_i$, $i\in[z]$. Then for arbitrary linear function $f: \mathbb{F}_q^2 \to \mathbb{F}_q$,  $f(c_i,c'_i)$, $i \in [n]$ are the shares of the scheme on message $f(m_1,m'_1)$ and keys $f(u_i,u'_i)$, $i \in [z]$.
\end{lemma}

\begin{proof}
	Follows from the linearity of (\ref{eq:shamir}).
\end{proof}

A secret sharing scheme allows secure and reliable storage of information, i.e., it can tolerate the loss of any $r$ shares as well as the exposure of any $z$ shares to an adversary. However, the problem of \emph{repair} is not addressed. Consider the situation that one or more shares are lost or unavailable, and so in order to maintain the same level of reliability one needs to reconstruct the lost shares. For example, in the application of storage, the $n$ shares are assigned to $n$ storage nodes, and in the situation of node failures, one wishes to repair the failures by reconstructing the shares originally assigned to the failed nodes. The repair problem can be easily solved if there is a trusted dealer, who can collect the available shares,  recompute the lost shares and reassign them to the failed or replacement nodes. However, the assumption of a trusted dealer responsible for centralized repair may not be practical in many applications. 

In this paper we study the situation that a trusted repair dealer is not available and the nodes holding the shares are responsible for carrying out the repair by themselves. A naive approach is to transmit the available shares to the failed node so that it can recompute its share. By doing so, however, information about the message may be leaked to the failed node. Therefore, the main question of interest is \emph{how to repair securely without a trusted dealer}. Below we utilize the ideas developed in secure multi-party computation \cite{BenOr88,Chaum88,Franklin92}, to design mechanisms to securely repair secret sharing schemes. We first formalize the notion of secure repair. Throughout the paper, for a vector such as $(m_1,\cdots,m_k)$ and an index set $I$, we denote $\{m_i:i\in I\}$ by $m_I$.

\begin{definition}\label{def:srepair}
	\emph{(Secure repair scheme)} Consider an $(n,k,r,z)$ secret sharing scheme and $n$ nodes such that node $i$ holds the share $c_i$. For any $e \in [n]$, and $I \subset [n]$, suppose that node $e$ fails and nodes in $I$ are available to help repairing node $e$. 	 
	 A secure repair scheme is a protocol of communication between the nodes, such that 1) the information sent by a node to other nodes is a function of the share it holds, its local coin flips, and the information it received from other nodes; 2) denote by $d_i$ all the information received by node $i$ from the protocol and denote by $u_i$ the result of coin flips at node $i$, then
	 \begin{itemize}
	 	\item (Repairability) $H(c_e|d_e) = 0$.
	 	\item (Security) $I(\bm{m}; c_A, u_A , d_  A) =0$, for all $A \subset [n]$, $|A| = z$.
	 \end{itemize}
\end{definition}

Note that Definition \ref{def:srepair} naturally extends the threat model of secret sharing, e.g., it maintains the security requirement that any $z$ nodes cannot learn any information about the message, during and after the repair process. 

\begin{construction} \label{con:sp1}
	\emph{(Generic secure repair)} Consider any $(n,k,r,z)$ secret sharing scheme, any $e \in [n]$, and any $I = \{i_1,\cdots,i_{|I|}\} \subset [n]$, $e \notin I$  such that there exists a linear function $f$ so that $  f(c_{i_1},c_{i_2},\cdots,c_{i_{|I|}}) = c_e$. Let $J=\{j_1, \cdots, j_{z+1}\}$ be an arbitrary subset of $[n]$ of size $z+1$. The secure repair process involves three steps:
	\begin{itemize}
		\item[1)] For each node $i \in I$, encode $c_i$ into $c_{i,1}, c_{i,2}, \cdots, c_{i,z+1}$ by a $(z+1,1,0,z)$ Shamir's scheme (in Construction \ref{con:sha} all nodes choose the same $\alpha_i$'s) and send $c_{i,k}$ to node $j_k \in J$.  
		\item[2)]  For each node $j \in J$, compute $c'_j = f(c_{i_1,j}, c_{i_2,j}, \cdots, c_{i_{|I|},j})$, and send $c'_j$ to node $e$.
		\item[3)] Node $e$ obtains $c_e$ by decoding the $(z+1,1,0,z)$ Shamir's scheme, regarding $c'_{j_1}, c'_{j_2}, \cdots, c'_{j_{z+1}}$ as the $z+1$ shares.
	\end{itemize}
\end{construction}

\begin{theorem}\label{thm:1}
	Construction \ref{con:sp1} is a secure repair scheme.
\end{theorem}
\begin{proof}
	We need to show that Construction \ref{con:sp1} meets the repairability and security requirements in Definition \ref{def:srepair}. By Lemma \ref{lem:linear}, $c'_{j_1}, c'_{j_2}, \cdots, c'_{j_{z+1}}$ are the shares of a $(z+1,1,0,z)$ Shamir's scheme that encodes $f(c_{i_1},c_{i_2},\cdots,c_{i_{|I|}}) = c_e$ as message. This proves repairability. We now focus on security. Let $A$ be an arbitrary set of nodes controlled by the adversary, with $|A|=z$. We consider two cases.
	
	\emph{Case 1: $e \notin A$.}  In this case $d_j = (c_{i_1,j}, c_{i_2,j}, \cdots, c_{i_{|I|},j})$ if $j \in J$, and $d_j = 0$ if $j \notin J$. Denote $c_{A,B} = \{c_{i,j}: i \in A, j \in B\}$, we have
	\begin{align}
		I(\bm{m}; c_A, u_A, d_A) & = I(\bm{m}; c_A, u_A, c_{I, J \cap A} ) \label{eq:start} \\
		& \stackrel{(a)}{=} I(\bm{m}; c_A, u_A, c_{I \backslash A, J \cap A} ) \nonumber \\
		& \stackrel{(b)}{=} I(\bm{m}; c_A, u_A | c_{I \backslash A, J \cap A}) + I(\bm{m}; c_{I \backslash A, J \cap A}) \nonumber \\
			& \stackrel{(c)}{\le} I(\bm{m}; c_A, u_A  | c_{I \backslash A, J \cap A}) + I(\bm{c}; c_{I \backslash A, J \cap A}) \nonumber \\
			& \stackrel{(d)}{=} I(\bm{m}; c_A, u_A | c_{I \backslash A, J \cap A}) \nonumber \\
			& \stackrel{(e)}{=} I(\bm{m}; c_A, u_A ) \nonumber \\
			& \stackrel{(f)}{=}  I(\bm{m}; c_A ) \nonumber \\
			& \stackrel{(g)}{=} 0 \label{eq:end}.
	\end{align}
	Here (a) is due to the fact that $c_{I \cap A, J}$ is a function of $c_A$ and $\bm{u}_A$; (b) follows from the chain rule; (c) follows from the data processing inequality and the Markov chain $\bm{m} \to \bm{c} \to c_{I \backslash A, J \cap A}$, i.e., $c_{I \backslash A, J \cap A}$ can be dependent on $\bm{m}$ only via $\bm{c}$; (d) follows from the fact that $c_{I \backslash A, J \cap A}$ are the shares of $|I \backslash A |$ independent $(z+1,1,0,z)$ secret sharing schemes and that for each scheme at most $|J \cap A| \le z$ of its shares are included; (e) follows from the fact that $(\bm{m}, c_A, u_A) \perp c_{I \backslash A, J \cap A}$, implied by (d); (f) follows from $(\bm{m}, c_A)\perp u_A$ ; and (g) follows from security of the secret sharing scheme being repaired.
	
	\emph{Case 2: $e \in A$}.  
	Since $|A|=z$ and $|J|=z+1$, $J \backslash A$ is not empty. Assume with out loss of generality that $j_1 \in J \backslash A$. Because $c'_{j_1}, c'_{j_2}, \cdots, c'_{j_{z+1}}$ are the shares of a $(z+1,1,0,z)$ Shamir's scheme that encodes $c_e$, it follows that $I(c_e; c'_{j_2}, c'_{j_3}, \cdots, c'_{j_{z+1}})=0$ and that there exists a linear function $g$ such that $g(c'_{j_1}, c'_{j_2}, \cdots, c'_{j_{z+1}}) = \sum_{k=1}^{z+1} g_k c'_{j_k}  =c_e$. This implies that $g_1 \ne 0$ and so $c'_{j_1} = ( c_e -  \sum_{k=2}^{z+1} g_k c'_{j_k}) g_1^{-1}$, namely,
	\begin{align}\label{eq:func}
		H(c'_{j_1} | c_e, c'_{J \backslash \{j_1\}})=0.
	\end{align}
	We have,
	\begin{align}
		I(\bm{m}; c_A, u_A, d_A) & = I(\bm{m}; c_A, u_A, c'_{J}, c_{I , A \cap J} ) \nonumber \\
		& \stackrel{(h)}{\le} I(\bm{m}; c_A, u_A, c'_J, c_{I , J\backslash\{j_1\}} ) \nonumber \\
		& \stackrel{(i)}{=} I(\bm{m}; c_A, u_A, c'_{J\backslash\{j_1\}}, c_{I , J\backslash\{j_1\}} ) \nonumber\\
		& \stackrel{(j)}{=} I(\bm{m}; c_A, u_A, c_{I , J\backslash\{j_1\}} ). \label{eq:con}
	\end{align}
	Here (h) follows from $A\cap J \subset J\backslash\{j_1\}$; (i) follows from (\ref{eq:func}); and (j) follows from the fact that $c'_{J\backslash\{j_1\}}$ is a function of $c_{I , J\backslash\{j_1\}}$.
	We continue the chain of inequality by treating (\ref{eq:con}) in a similar way as Case 1. Namely, applying an argument similar to that of (\ref{eq:start}) - (\ref{eq:end}), we have
	\begin{align}
		I(\bm{m}; c_A, u_A, d_A) & \le I(\bm{m}; c_A, u_A, c_{I , J\backslash\{j_1\}} ) \nonumber \\& = I(\bm{m};c_A, u_A, c_{I \backslash A, J\backslash\{j_1\}} ) \nonumber \\
		& = I(\bm{m}; c_A, u_A | c_{I \backslash A, J\backslash\{j_1\}}) + I(\bm{m};c_{I \backslash A, J\backslash\{j_1\}}) \nonumber \\
		& \le I(\bm{m}; c_A, u_A  | c_{I \backslash A, J\backslash\{j_1\}}) + I(\bm{c}; c_{I \backslash A, J\backslash\{j_1\}}) \nonumber \\
		& = I(\bm{m}; c_A, u_A | c_{I \backslash A, J\backslash\{j_1\}})  \nonumber \\
		& = I(\bm{m}; c_A, u_A ) \nonumber \\
		& =  I(\bm{m}; c_A ) \nonumber \\
		& = 0. \nonumber
	\end{align}
	The proof is complete.
\end{proof}
We remark that Construction \ref{con:sp1} is a generic scheme that can securely repair any linear secret sharing scheme. Particularly, it does not require modifying the secret sharing scheme. In a sense this suggests that secure repair ``comes for free'' without needing to compromise other aspects of efficiency of the scheme. In comparison, the secure regenerating codes in \cite{Pawar11,Shah:2011bz,Rawat:2014cu,Tandon16} allow secure repair at the cost of reducing rate. We also remark that multiple failures can be repaired securely by invoking Construction \ref{con:sp1}  multiple times.

We analyze the repair bandwidth, i.e., the total amount of information that is communicated during the repair process. In Step 1, at most $|I|(z+1)$ symbols are transmitted and in Step 2, at most $z+1$ symbols are transmitted. Therefore the total repair bandwidth is at most $(|I|+1)(z+1)$ symbols, which is approximately $z+1$ times of the non-secure repair bandwidth $|I|$.

\section{Reducing the secure repair bandwidth} \label{sec:3}
While Construction \ref{con:sp1} provides a generic approach to repair secret sharing schemes securely, it incurs a large overhead in the repair bandwidth. In this section we propose an improved generic secure repair scheme with a significantly better repair bandwidth. The main idea is that, instead of repairing one single share/symbol at a time, we repair multiple shares together in parallel, and therefore amortizing the communication overhead over the multiple shares. For this to work we need every node to store multiple shares, which is typically the case because the amount of information to be stored (e.g., a  file) usually exceeds the amount of information that can be stored by a single secret sharing scheme. Therefore the file will be split and stored by multiple independent instances of a secret sharing scheme, resulting in multiple shares to be assigned to a node. In the reminder of the paper we assume that there are  enough shares in the failed node to be repaired. Then, the main improvement is that in the first round of the repair scheme, rather than using a low rate $(z+1,1,0,z)$ secret sharing scheme, we use a high rate $(n,n-z,0,z)$ scheme. This allows one to repair $n-z$ shares in parallel and reduce the amortized overhead in the repair bandwidth (which are the $z$ keys in the secret sharing schemes of the first round) by $n-z$ times.

Formally, we assume that each node stores $n-z$ shares from $n-z$ independent instances of a secret sharing scheme. We use superscripts to index instances, e.g., $\bm{m}^{(i)} = (m_1^{(i)},\cdots,m_k^{(i)})$ is the message encoded by the $i$-th instance. In the first round of repair we use a high rate secret sharing scheme defined in Construction \ref{con:ramp}, which is a generalization of Shamir's scheme to the case of $k>1$.

\begin{construction} \label{con:ramp}
	\emph{(Ramp version of Shamir's scheme \cite{Franklin92, HB16})} For any $n$, $r$, $z$ such that $n>r+z$, let $k=n-r-z$ and $\mathbb{F}_q$ be a finite field of size $q>n$. Let $u_i$, $i \in [z]$ be i.i.d. uniformly distributed over $\mathbb{F}_q$ and let $\alpha_i$, $i\in[n]$ be arbitrary distinct non-zero elements of $\mathbb{F}_q$. The shares corresponding to message $\bm{m}=(m_1,m_2,\cdots,m_k)$ are
	\begin{align}
		(c_1, c_2, \cdots,c_n) = (m_1, \cdots, m_k, u_1, \cdots, u_z) \left[ \begin{array}{cccc}
			1 & 1 & \cdots & 1\\
			\alpha_1 & \alpha_2 & \cdots & \alpha_n\\
			\vdots & \vdots & \vdots & \vdots\\
			\alpha_1^{z+k-1} & \alpha_2^{z+k-1} & \cdots & \alpha_n^{z+k-1}
		\end{array} \right]. \nonumber
	\end{align}
\end{construction}
Construction \ref{con:ramp} is an $(n,k=n-r-z,r,z)$ secret sharing scheme \cite{HB16}. 

\begin{construction} \label{con:sp2}
		\emph{(Bandwidth-efficient secure repair)} Consider any $(n,k,r,z)$ secret sharing scheme, any $e \in [n]$, and any $I = \{i_1,\cdots,i_{|I|}\} \subset [n]$, $e \notin I$  such that there exists a linear function $f$ so that $  f(c_{i_1},c_{i_2},\cdots,c_{i_{|I|}}) = c_e$. The secure repair process involves three steps:
	\begin{itemize}
		\item[1)] For each node $i \in I$, encode $c_i^{(1)}, c_i^{(2)}, \cdots, c_i^{(n-z)}$ into $c_{i,1}, c_{i,2}, \cdots, c_{i,n}$ by a $(n,n-z,0,z)$ scheme according to Construction \ref{con:ramp} (all nodes should choose the same $\alpha_i$'s) and send $c_{i,j}$ to node $j$.  
		\item[2)]  For each node $j \in [n]$, compute $c'_j = f(c_{i_1,j}, c_{i_2,j}, \cdots, c_{i_{|I|},j})$, and send $c'_j$ to node $e$.
		\item[3)] Node $e$ obtains $c^{(1)}_e, c^{(2)}_e, \cdots, c^{(n-z)}_e$ by decoding the $(n,n-z,0,z)$ scheme, regarding $c'_1, c'_2, \cdots, c'_n$ as the $n$ shares.
	\end{itemize}
\end{construction}

\begin{theorem} \label{thm:2}
	Construction \ref{con:sp2} is a secure repair scheme.
\end{theorem}
\begin{proof}
	Similar to Theorem \ref{thm:1}, repairability follows from the linearity of Construction \ref{con:ramp}, which implies that $c'_{[n]}$ are the shares of a $(n,n-z,0,z)$ secret sharing scheme that encodes $ ( f(c_{i_1}^{(1)}, c_{i_2}^{(1)}, \cdots, c_{i_{|I|}}^{(1)} ),\cdots, f(c_{i_1}^{(n-z)}, c_{i_2}^{(n-z)}, \cdots, c_{i_{|I|}}^{(n-z)} ) ) = (c_e^{(1)}, \cdots, c_e^{(n-z)})$ as message. Focusing on security, let $A \subset [n]$, $|A|=z$ be an arbitrary set of nodes controlled by the adversary, then by the property of Construction \ref{con:ramp} it follows that $c_e^{[n-z]} \perp c'_A$ and $H(c'_A) = z$. We have
	\begin{align}
		H(\bm{c}'_{[n]\backslash A} | c_e^{[n-z]}, c'_A ) & \stackrel{(a)}{=} H(c_e^{[n-z]}, c'_{[n]}) - H(c_e^{[n-z]}, c'_A) \label{eq:h0begin}\\
		& \stackrel{(b)}{\le} H(c_e^{[n-z]})+z - H(c_e^{[n-z]}, c'_A) \nonumber \\
		& \stackrel{(c)}{=} H(c_e^{[n-z]})+z - H(c_e^{[n-z]}| c'_A) - H( c'_A) \nonumber \\
		& \stackrel{(d)}{=} H(c_e^{[n-z]})+z - H(c_e^{[n-z]}) - H( c'_A) \nonumber \\
		& = z - H( c'_A) \nonumber \\
		& \stackrel{(e)}{=} 0. \label{eq:h0}
	\end{align}
	Here (a) and (c) follows from the chain rule; (b) follows from the fact that $c'_{[n]}$ is a function of  $c_e^{[n-z]}$ and $z$ random keys; (d) follows from  $ c_e^{[n-z]} \perp c'_A$ and (e) follows from $H( c'_A) = z$. 
	
	Consider the case that $e \in A$, we have
	\begin{align}\label{eq:init2}
	I(\bm{m}^{[n-z]}; c_A^{[n-z]}, u_A, d_A) & = I(\bm{m}^{[n-z]}; c_A^{[n-z]}, u_A, c'_{[n]}, c_{I , A} )\\
	   & \stackrel{(f)}{=} I(\bm{m}^{[n-z]}; c^{[n-z]}_A, u_A, c'_A, c_{I , A} ) \nonumber \\
	& \stackrel{(g)}{=} I(\bm{m}^{[n-z]}; c_A^{[n-z]}, u_A,  c_{I , A} ), \label{eq:first2} \\
	& \stackrel{(h)}{=} I(\bm{m}^{[n-z]}; c^{[n-z]}_A, u_A, c_{I \backslash A, A} ) \nonumber \\
	& \stackrel{(i)}{=} I(\bm{m}^{[n-z]}; c^{[n-z]}_A, u_A | c_{I \backslash A, A}) + I(\bm{m}^{[n-z]}; c_{I \backslash A, A}) \nonumber \\
	& \stackrel{(j)}{\le} I(\bm{m}^{[n-z]}; c^{[n-z]}_A, u_A | c_{I \backslash A, A}) + I(\bm{c}^{[n-z]}; c_{I \backslash A, A}) \nonumber \\
	& \stackrel{(k)}{=} I(\bm{m}^{[n-z]}; c^{[n-z]}_A, u_A | c_{I \backslash A, A}) \nonumber \\
	& \stackrel{(l)}{=} I(\bm{m}^{[n-z]}; c^{[n-z]}_A, u_A ) \nonumber \\
	& \stackrel{(m)}{=}  I(\bm{m}^{[n-z]}; c_A^{[n-z]} ) \nonumber \\
	& \stackrel{(n)}{=} 0, \label{eq:last2}
	\end{align}
	where (f) follows from (\ref{eq:h0}); (g) follows from the fact that $c'_A$ is a function of $c_{I,A}$; (h) follows from the fact that $c_{A,A}$ is a function of $c_A^{[n-z]}$ and $u_A$; (i) follows from the chain rule; (j) follows form the Markov chain $\bm{m}^{[n-z]} \to \bm{c}^{[n-z]} \to c_{I \backslash A, A}$ and the data processing inequality; (k) follows from the fact that $c_{I \backslash A, A}$ are the shares of $|I \backslash A |$ independent $(n,n-z,0,z)$ secret sharing schemes and that for each scheme only $|A| = z$ of its shares are included; (l) follows from the fact that $(\bm{m}^{[n-z]}, c_A^{[n-z]}, u_A) \perp c_{I \backslash A, A}$, implied by (k); (m) follows from $(\bm{m}^{[n-z]}, c_A^{[n-z]})\perp u_A$; and (n) follows from security of the secret sharing scheme being repaired. 
	
	For the case that $e \notin A$, we have $I(\bm{m}^{[n-z]}; c_A^{[n-z]}, u_A, d_A) = I(\bm{m}^{[n-z]}; c^{[n-z]}_A, u_A,  c_{I , A} )$, which can be treated in the same way as (\ref{eq:first2}) - (\ref{eq:last2}). The proof is complete.
\end{proof}

In Step 1, at most $|I|n$ symbols are communicated and in Step 2, at most $n$ symbols are communicated. Therefore the total repair bandwidth is at most $(|I|+1)n$ symbols, for repairing $n-z$ symbols. The normalized repair bandwidth to repair each symbol is at most $\frac{(|I|+1)n}{n-z}$ symbols. In the case that $n$ dominates $z$, the normalized repair bandwidth approaches $|I|+1$ symbols. Note that $|I|$ is the non-secure repair bandwidth and a trivial lower bound on the secure repair bandwidth. Therefore when $n$ dominates $z$ (e.g., the high rate case), the secure repair bandwidth of Construction \ref{con:sp2} is essentially optimal. Specifically, it is essentially the same as the non-secure repair bandwidth, implying that we can have secure repair essentially for free, even in terms of repair bandwidth.

\section{Vector linear secure repair} \label{sec:4}
The secure repair schemes in Constructions \ref{con:sp1} and \ref{con:sp2} deal with scalar secret sharing schemes, i.e., schemes that are linear over a finite field and such that each share is an element of the field. A more general class of secret sharing schemes are vector linear secret sharing schemes, also referred to as array schemes.  A vector linear $(n,k,r,z)$ secret sharing scheme over $\mathbb{F}_q^t$ is a randomized function that maps (encodes) a message $\bm{m}=(m_1,\cdots,m_k)$ of $k$ symbols over $\mathbb{F}_q^t$ to $n$ shares $\bm{c}=(c_1,\cdots,c_n)$ over $\mathbb{F}_q^t$, such that the encoding function is linear over $\mathbb{F}_q$ and that the same reliability and security requirements as before are met. We denote $m_i = (m_{i,1}, m_{i,2}, \cdots, m_{i,t})$, where $m_{i,j} \in \mathbb{F}_q$, for $i \in [k], j \in [t]$. Similarly we denote $c_{i}=(c_{i,1}, c_{i,2}, \cdots, c_{i,t})$, for $i \in [n]$. Note that scalar schemes are special cases of vector schemes with $t=1$. 

Many efficient secret sharing schemes, e.g., schemes with efficient decoding bandwidth \cite{HB16,Huang17}, schemes with efficient computation \cite{HB16b,HB17sRAID}, and schemes with efficient repair bandwidth \cite{Shah:2011bz,Huang17}, are intrinsically vector linear. In this section we extend our secure repair framework to vector linear schemes. This is especially interesting because it allows us to leverage the property of secret sharing schemes with efficient (non-secure) repair bandwidth, i.e., secure regenerating codes, to further reduce the (secure) repair bandwidth.

We remark that existing secure regenerating codes can be classified into two categories: codes with non-secure repair and codes with secure repair. Secure regenerating codes with non-secure repair focus on reducing the repair bandwidth without worrying about the security of the repair process. In this case one can think of having a trustworthy repair dealer that will reconstruct the lost share and forward it to the failed node. As remarked previously, during the repair process the dealer may gain information about the message and therefore has  to be trustworthy.
In comparison, regenerating codes with secure repair, by code design, guarantee that such a dealer will not  learn any information about the message. This in fact removes the need for the dealer to be trustworthy and the failed node can act as the dealer. In this sense, secure regenerating codes with secure repair naturally admit a secure repair scheme that meets Definition \ref{def:srepair}. Particularly, the secure repair scheme is a simple ``one-round'' scheme in the sense that  the helper nodes will transmit information to the failed node but they themselves do not need to receive information from other nodes. Unfortunately, one-round secure repair comes at a high cost in rate and codes with non-secure repair generally have a much better rate as well as repair bandwidth when normalized by rate than codes with secure repair \cite{Pawar11, Tandon16}. Our main result in this section implies that this trade-off between rate and secure repair is not necessary: we can apply our generic approach to secure regenerating codes with non-secure repair to achieve secure repair, a good rate, and a good repair bandwidth. The only cost is that the repair process now involves two rounds instead of one round.

\begin{construction} \label{con:sp3}
	\emph{(Vector linear secure repair)} Consider any vector linear $(n,k,r,z)$ secret sharing scheme over $\mathbb{F}_q^t$, any $e \in [n]$, and any $I = \{i_1,\cdots,i_{|I|}\} \subset [n]$, $e \notin I$  such that there exists $J \subset [t]$ and a linear function $f$ over $\mathbb{F}_q$ that takes $c_{i,j}$, $i \in I$, $j \in J$ as input and outputs $c_e = (c_{e,1}, c_{e,2}, \cdots, c_{e,t})$. The secure repair process involves three steps:
	\begin{itemize}
		\item[1)] For each node $i \in I$, and $j \in J$, encode $c_{i,j}^{(1)}, c_{i,j}^{(2)}, \cdots, c_{i,j}^{(n-z)}$ into $c_{i,j,1}, c_{i,j,2}, \cdots, c_{i,j,n}$ by a $(n,n-z,0,z)$ scheme according to Construction \ref{con:ramp} (all nodes choosing the same $\alpha_1, \alpha_2 \cdots, \alpha_n$) and send $c_{i,j,k}$ to node $k$.  
		\item[2)]  For each node $k \in [n]$, compute $(c'_{k,1}, c'_{k,2}, \cdots, c'_{k,t}) = f(c_{i,j,k})_{i \in I, j\in J}$, and send $c'_{k,j}$, $j \in [t]$ to node $e$.
		\item[3)] For $j \in [t]$, node $e$ obtains $c^{(1)}_{e,j}, c^{(2)}_{e,j}, \cdots, c^{(n-z)}_{e,j}$ by decoding the $(n,n-z,0,z)$ scheme, regarding $c'_{1,j}, c'_{2,j}, \cdots, c'_{n,j}$ as the $n$ shares.
	\end{itemize}
\end{construction}

\begin{theorem} \label{thm:3}
	Construction \ref{con:sp3} is a secure repair scheme.
\end{theorem}

\begin{proof}
	As in Theorem \ref{thm:2}, repairability follows from the linearity of Construction \ref{con:ramp}, which implies that $c'_{1,j}, c'_{2,j}, \cdots, c'_{n,j}$ are the shares of a $(n,n-z,0,z)$ secret sharing scheme that encodes $c_{e,j}^{(1)}, c_{e,j}^{(2)}, \cdots, c_{e,j}^{(n-z)}$ as message, for $j \in [t]$. We now turn to security, and follow a similar flow as Theorem \ref{thm:2}. Let $A \subset [n]$, $|A|=z$ be an arbitrary set of nodes controlled by the adversary, then by the property of Construction \ref{con:ramp} it follows that $c_e^{[n-z]} \perp c'_{A,j}  = 0$ and $H(c'_{A,j}) = z$, for $j \in [t]$. We have, for $j \in [t]$,
	\begin{align}
		H( c'_{[n]\backslash A, j} | c_{e,j}^{[n-z]}, c'_{A,j} ) & = H(c_{e,j}^{[n-z]}, c'_{[n],j}) - H(c_{e,j}^{[n-z]}, c'_{A,j}) \nonumber \\
		& \le H(c_{e,j}^{[n-z]})+ z - H(c_{e,j}^{[n-z]}, c'_{A,j}) \nonumber \\
		& = H(c_{e,j}^{[n-z]}) + z - H(c_{e,j}^{[n-z]}|c'_{A,j}) - H(c'_{A,j}) \nonumber \\
		& = H(c_{e,j}^{[n-z]})+ z - H(c_{e,j}^{[n-z]}) - H(c'_{A,j}) \nonumber \\
		& = z - H(c'_{A,j}) \nonumber \\
		& = 0, \label{eq:h00}
	\end{align}
	where the justification for the steps are similar to that of (\ref{eq:h0begin}) - (\ref{eq:h0}). 
	(\ref{eq:h00}) implies that 
	\begin{align}\label{eq:rh0}
		H( c'_{ [n]\backslash A, [t] } | c_{e,[t]}^{[n-z]}, c'_{A,[t]} ) = 0.
	\end{align}
	Now consider the case that $e \in A$, we have
	\begin{align}
		I(\bm{m}^{[n-z]}; c_{A,[t]}^{[n-z]}, u_A, d_A) & = I(\bm{m}^{[n-z]}; c_{A,[t]}^{[n-z]}, u_A, c'_{[n],[t]}, c_{I , J, A} ) \nonumber \\
		& \stackrel{(a)}{=} I(\bm{m}^{[n-z]}; c^{[n-z]}_{A,[t]}, u_A, c'_{A,[t]}, c_{I , J, A} ) \nonumber \\
		& = I(\bm{m}^{[n-z]}; c_{A,[t]}^{[n-z]}, u_A,  c_{I , J, A} ), \label{eq:first3} \\
		& = I(\bm{m}^{[n-z]}; c^{[n-z]}_{A,[t]}, u_A, c_{I \backslash A, J, A} ) \nonumber \\
		& = I(\bm{m}^{[n-z]}; c^{[n-z]}_{A,[t]}, u_A | c_{I \backslash A, J, A}) + I(\bm{m}^{[n-z]}; c_{I \backslash A, J, A}) \nonumber \\
		& \le I(\bm{m}^{[n-z]}; c^{[n-z]}_{A,[t]}, u_A | c_{I \backslash A, J, A}) + I(\bm{c}^{[n-z]}; c_{I \backslash A, J, A}) \nonumber \\
		& \stackrel{(b)}{=} I(\bm{m}^{[n-z]}; c^{[n-z]}_{A,[t]}, u_A | c_{I \backslash A, J, A}) \nonumber \\
		& = I(\bm{m}^{[n-z]}; c^{[n-z]}_{A,[t]}, u_A ) \nonumber \\
		& =  I(\bm{m}^{[n-z]}; c_{A,[t]}^{[n-z]} ) \nonumber \\
		& = 0, \label{eq:last3}
	\end{align}
	where (a) follows from (\ref{eq:rh0}); (b) follows from the fact that $c_{I \backslash A ,J,A}$ are the shares of $|I \backslash A | \cdot |J|$ independent $(n,n-z,0,z)$ secret sharing schemes and that for each scheme only $|A| = z$ of its shares are included; and the remaining equalities/inequalities are similar to (\ref{eq:init2}) - (\ref{eq:last2}).
	
	For the case that $e \notin A$, we have $I(\bm{m}^{[n-z]}; c_{A,[t]}^{[n-z]}, u_A, d_A) = I(\bm{m}^{[n-z]}; c^{[n-z]}_{A,[t]}, u_A,  c_{I , J, A} )$, which can be treated in the same way as (\ref{eq:first3}) - (\ref{eq:last3}). The proof is complete.
\end{proof}

Consider the repair bandwidth of the scheme. In Step 1, at most $n|I||J|$ symbols (over $\mathbb{F}_q$) are transmitted and in Step 2, at most $nt$ symbols are transmitted. Therefore, the total repair bandwidth is at most $(|I||J|+t)n$ symbols, for repairing $(n-z)t$ symbols. The normalized repair bandwidth to repair each symbol is at most $\frac{(|I||J|+t)n}{(n-z)t}$ symbols. In the case that $n$ dominates $z$, the normalized repair bandwidth approaches $\frac{|I||J|}{t}+1$. Note that the normalized non-secure repair bandwidth is $\frac{|I||J|}{t}$, and therefore in this case the the secure repair bandwidth of Construction \ref{con:sp3} is essentially optimal and is almost the same as the non-secure repair bandwidth.

The MSR secure regenerating codes in \cite{Shah:2011bz,Huang17} have optimal rate as well as optimal non-secure repair bandwidth (among rate-optimal schemes). Specifically, the rate of the scheme is $\frac{n-k-z}{n}$, and that for $|I|$ helper nodes to non-securely repair a failed node, each helper node will transmit $1/(1+|I|-k-z)$ fraction of the symbols it stores, i.e., $|J| = \frac{t}{1+|I|-k-z}$. By applying Construction \ref{con:sp3} to these codes, we obtain schemes with optimal rate and low secure repair bandwidth. In the next section we will show that the secure repair bandwidth is in fact optimal up to a small constant factor.



\section{Lower bound on secure repair bandwidth} \label{sec:5}
The bandwidths of Construction \ref{con:sp2} and Construction \ref{con:sp3} are significantly better than Construction \ref{con:sp1}. A natural question is whether it is possible to do even better, or in other words, what is a lower bound on the secure repair bandwidth. As we previously remarked, when $n$ dominates $z$, the  bandwidths of Constructions \ref{con:sp2} and \ref{con:sp3} approach the non-secure repair bandwidth, which is a naive lower bound. Therefore in this case the bandwidths of Constructions \ref{con:sp2} and \ref{con:sp3} are asymptotically optimal. In this section, we prove a stronger lower bound on the secure repair bandwidth and show that the  bandwidths of Constructions \ref{con:sp2} and \ref{con:sp3} are  optimal for all parameters up to a constant factor of 2, as long as the secret sharing scheme being repaired is rate-optimal. 

Assume that a trustworthy repair dealer is available. The dealer will receive information from the helper nodes, evaluate a \emph{repair function} that outputs the lost share, and reassign the share. In this case, the repair bandwidth is the size of the input to the repair function plus the size of the lost share. Now consider the situation that a trustworthy dealer is not available and a secure repair scheme is used for repair. The secure repair scheme essentially is a method to evaluate the repair function (e.g., $f$ in Constructions \ref{con:sp1}, \ref{con:sp2} and \ref{con:sp3}) securely at the failed node, and the repair bandwidth again depends on the size of the input to the repair function. The repair function is an intrinsic component of the secret sharing scheme and the size of the input can be minimized by carefully designing the secret sharing scheme, e.g., \cite{Shah:2011bz,Huang17}. Refer to the size of the input to the repair function as the non-secure repair bandwidth of a secret sharing scheme. Below we prove a lower bound on the repair bandwidth of secure repair schemes, given the non-secure repair bandwidth of the secret sharing scheme.

\begin{theorem} \label{thm:lbnd}
	For any rate-optimal $(n,k=n-r-z,r,z)$ secret sharing scheme, let $W$ be the non-secure repair bandwidth of the scheme, then a secure repair scheme requires a bandwidth of at least $\frac{(n-1)W}{2(n-z-1)}$. 
\end{theorem}

\begin{proof}
	Note that $k=n-r-z$ implies that the scheme is rate-optimal \cite{HB16}. Let the message $\bm{m}= (m_1,m_2,\cdots,m_k)$ be uniformly distributed. Then for any $I \subset [n]$, $|I|=k+z$ and $J \subset I$, $|J|=z$, by the security and the decodability of the scheme we have $I( \bm{m} ; c_{J}) = 0 $ and $I( \bm{m} ; c_I)=H(\bm{m})=k$. It follows that
	\begin{align}
		I( \bm{m} ; c_{I \backslash J} | c_J) & = H(\bm{m} | c_J) - H( \bm{m} |c_I ) \nonumber \\
		 & = H(\bm{m}) - H( \bm{m} |c_I ) \nonumber \\
		 & = I(\bm{m}; c_I)  \nonumber \\
		 & = k. \label{eq:cijk}
	\end{align}
	Since $|I \backslash J| = k$, $H(c_{I \backslash J}) \le k$, and hence (\ref{eq:cijk}) implies that $H(c_{I \backslash J}) = k$ and  $c_{I \backslash J} \perp c_J$ and that
	\begin{align} \label{eq:cIJ0}
		H( c_{I \backslash J} | \bm{m}, c_J) =0.
	\end{align}
	Therefore among the $n$ shares of the secret sharing scheme, any $|I \backslash J|=k$ shares are uniformly distributed and that any $|J|=z$ shares are independent of any other $k$ shares. This in turn implies that any $k+z$ shares are uniformly distributed, i.e., 
	\begin{align} \label{eq:uniform}
		H(c_I) = k+z.
	\end{align}
	
	Assume that $c_e$ is lost, and for $i \in [n]\backslash\{e\}$, let $w_i$ be the information sent by node $i$ to node $e$ for non-secure repair, namely, the input signal to the repair function from node $i$ (with the convention that $w_i=0$ if node $i$ does not participate in the repair). Then $w_i$ is a function of $c_i$ and $\sum_{i \in [n]\backslash \{e\}} H(w_i) = W$. Now consider any secure repair protocol, and for $i \in [n]\backslash\{e\}$, $j \in [n]$, let $v_{i,j}$ be the set of signals that are sent to node $j$ by node $i$ or sent to node $i$ by node $j$ during the protocol (with the convention that $v_{i,i}=\emptyset$). Then $w_i$ must be a function of the signals incoming to and outgoing from node $i$, namely, $H( w_i | v_{i,[n]}) = 0$, implying that
	\begin{align} \label{eq:iwv}
		I(w_i ; v_{i,[n]}) = H(w_i).
	\end{align}
	
	Let $A$ be an arbitrary set of nodes controlled by the adversary such that $i \notin A$, $|A|=z$, and let $B$ be an arbitrary set of nodes such that $i \in B$, $|B|=k$, $A \cap B = \emptyset$. We have
	\begin{align}
		I( w_i ; v_{i,A} ) & = H(w_i) - H(w_i | v_{i,A}) \nonumber \\
		& \stackrel{(a)}{=} H(w_i | c_A) - H(w_i | v_{i,A}) \nonumber \\
		& \le  H(w_i | c_A) - H(w_i | v_{i,A}, c_A) \nonumber \\
		& = I(w_i ; v_{i,A} | c_A) \nonumber \\
		& \stackrel{(b)}{\le} I(c_i ; v_{i,A} | c_A) \nonumber \\
		& \le I(c_B ; v_{i,A} | c_A) \nonumber \\
		& \stackrel{(c)}{=} I(\bm{m}; v_{i,A} | c_A) \nonumber \\
		& \stackrel{(d)}{=} I(\bm{m}; v_{i,A} | c_A) + I(\bm{m} ; c_A) \nonumber \\
		& = I(\bm{m} ; v_{i,A}, c_A) \nonumber \\
		& \stackrel{(e)}{=} 0. \label{eq:iwv0}
	\end{align}
	Here (a) follows from the fact that $w_i$ is a function of $c_i$ and by (\ref{eq:uniform}), $c_i \perp c_A$; (b) follows from the data processing inequality and the fact that $w_i$ is a function of $c_i$; (c) follows from the data processing inequality and (\ref{eq:cIJ0}), i.e., $c_B$ is a function of $\bm{m}$  given $c_A$; (d) follows from the security of the secret sharing scheme; and (e) follows from the security of the repair scheme.
	Let 
	\begin{align}
		A^* = \underset{A \subset [n]\backslash \{i\}, |A|=z}{\text{argmax}} \sum_{l \in A} H(v_{i,l}), \nonumber
	\end{align}
	and let $\bar{A}^* = [n]\backslash(\{i\} \cup A^*)$, then for $j \in \bar{A}^*$ and $j^* \in A^*$, $H(v_{i,j}) \le H(v_{i,j^*})$. We have 
	\begin{align*}
		H(v_{i, \bar{A}^* }) & \ge I( w_i ; v_{i, \bar{A}^* } | v_{i,A}  )\\
		& \stackrel{(f)}{=} I( w_i ;  v_{i, \bar{A}^* }  | v_{i,A}  ) + I( w_i ; v_{i,A} )\\
		& = I(w_i ; v_{i,[n]})\\
		& \stackrel{(g)}{=} H(w_i),
	\end{align*}
	where (f) follows from (\ref{eq:iwv0}) and (g) follows from (\ref{eq:iwv}).
	Therefore there exist $j \in \bar{A}^*$ such that $H(v_{i,j}) \ge H(w_i)/|\bar{A}^*|$ and so for $j^* \in A^*$, $H(v_{i,j^*}) \ge H(w_i)/(n-z-1)$. Therefore the amount of information transmitted and received by node $i$ is lower bounded by
	\begin{align}
		\sum_{j \in [n]} H(v_{i,j}) & \ge H(v_{i, \bar{A}^* }) + |A^*|\frac{H(w_i)}{n-z-1} \nonumber \\
		& = \frac{(n-1)H(w_i)}{n-z-1} \label{eq:lbndv}.
	\end{align}
	Summing (\ref{eq:lbndv}) over all $i \in [n]\backslash \{e\}$, it follows that the amount of information transmitted and received by nodes in $[n] \backslash \{e\}$ is at least $\frac{(n-1) W}{n-z-1}$. Since the amount of communication is counted exactly twice, i.e., when information is transmitted and when it is received, the repair bandwidth of the scheme is lower bounded by $\frac{(n-1) W}{2(n-z-1)}$. This completes the proof.
\end{proof}
	The bandwidths of Constructions \ref{con:sp2} and $\ref{con:sp3}$ are upper bounded by $\frac{(W+1)n}{n-z}$, and therefore are optimal up to a factor of approximately 2 by Theorem \ref{thm:lbnd}.

\section{Concluding Remarks}
This paper studies the problem of repairing lost shares of a secret sharing scheme without a trustworthy repair dealer. We design generic repair schemes that can securely repair \emph{any} (scalar or vector) linear secret sharing schemes. We prove a lower bound on the repair bandwidth of secure repair schemes and show that the proposed secure repair schemes achieve the optimal repair bandwidth up to a small constant factor when $n$ dominates $z$, or when the secret sharing scheme being repaired has optimal rate.

An interesting open problem is to study the secure repair bandwidth under the general repair model when the secret sharing scheme being repaired is not rate-optimal. More generally, while the tradeoff between repair bandwidth and rate has attracted significant interests under the one-round repair model, under the general repair model whether a tradeoff exists or not and how to characterize it remain open. Another interesting open problem is to study secure repair in the presence of active adversarial nodes that may deviate from the prescribed repair protocol.

\bstctlcite{IEEEexample:BSTcontrol}
\bibliographystyle{IEEEtranS}
\bibliography{securerepair}
\end{document}